\documentclass[pra,twocolumn,superscriptaddress]{revtex4}

\usepackage{mathtools}
\usepackage{float}
\usepackage{graphicx}
\usepackage{epsfig}

\usepackage{dcolumn}
\usepackage{amsmath,amssymb,amsthm}
\usepackage{bm}
\usepackage{verbatim}
\usepackage{natbib}

\usepackage{color}

\newcommand{\bra}[1]{\left\langle #1 \right|}
\newcommand{\ket}[1]{\left| #1 \right\rangle}

\newcommand{\ketbra}[2]{\left|#1\middle\rangle\middle\langle#2\right|}

\newcommand{\norm}[1]{\left\|#1\right\|}
\newcommand{\abs}[1]{\left|#1\right|}

\newcommand{\id}{\mathbb{I}}
\newcommand{\M}[1]{\mathcal{#1}}
\newcommand{\B}[1]{\mathbb{#1}}
\newcommand{\Tr}{\text{Tr}}
\newcommand{\N}[1]{\Vert #1 \Vert}
\newcommand{\rhoAB}{\rho_{AB}\in\M{D}(\B{C}_A\otimes\B{C}_B)}
\newcommand{\apos}{\textsc{\char13}}

\newtheorem{teo}{Theorem}
\newtheorem{lemma}[teo]{Lemma}
\newtheorem{defi}[teo]{Definition}

\newtheorem{prop}[teo]{Proposition}

\newtheorem{remark}[teo]{Remark}

\begin{document}
\title{Koashi-Winter relation for $\alpha$-Renyi entropies} 
\author{Tiago Debarba} 
\email{debarba@utfpr.edu.br}
\affiliation{Universidade Tecnol\'ogica Federal do Paran\'a (UTFPR), Campus Corn\'elio Proc\'opio. 
R. Alberto Carazzai, 1640, Corn\'elio Proc\'opio, PR, 86300-000 - Brazil.}

\date{\today}
\begin{abstract}
This work presents a generalization of the Koashi-Winter relation for $\alpha$-Renyi entropies. 
This result is based on the Renyi\apos s entropy version of quantum 
Jensen Shannon divergence. By means of this definition, a classical correlations 
quantifier $C_{\alpha}(\rho_{AB}) = \sup_{\xi_{AB}^{M_B}} Q_{\alpha}(\xi_{AB}^{M_B})$ is proposed, where 
the optimization is taken over the ensembles $\xi_{AB}^{M_B}$ 
created by the outputs of the local measurement 
process. The main result is applied to the capacity of a quantum classical channel over a 
tripartite pure state $\psi_{ABE}$, that is rated above in function of the probability of success to discriminate 
the states in the ensemble $\xi_{AE}^{M_E}$, created by the local dephasing over partition $E$, and the asymptotic 
log generalized robustness of partition $AB$. Some analytical results are calculated for classical correlations and 
entanglement of formation.

\end{abstract}
\maketitle

\section{Introduction}

Quantum correlations are intrinsically related to quantum information theory, 
 as resources for quantum information protocols \cite{mother}. The quantification of quantum 
information rates is performed by a quantum version of entropic measures, especially by von Neumann entropy 
\begin{equation}
S(\rho) = - \Tr(\rho \log \rho),
\end{equation}
and its related measures \cite{wildebook},  
although there are generalizations of von Neumann entropy, and its related measures, by means 
of Tsallis\apos entropy \cite{tsallis88} and Renyi\apos s entropy \cite{renyi61}.  
The quantum version of Renyi\apos s entropy, in the context of quantum correlations 
and quantum information, has been a theme of intense investigation in the last few years 
\cite{mosonyi2011,matos12,muller2013,fehr2014,mosonyi2015,datta15,berta15,berta15b,wilde15,wilde16b,wilde16,hermes2017}. 
This is the context in which this work is inserted.

The main result of this work is the generalization of the Koashi - Winter relation for 
Renyi\apos s entropy. This result is based on the Renyi\apos s entropy version of quantum 
Jensen Shannon Divergence (QJSD) \cite{fehr2014}. By means of Renyi\apos s QJSD 
and the distinguishability of quantum states, 
a quantifier of classical correlations is introduced:  $C_{\alpha}$, for $\alpha\in (0,1)$. 
The properties for $C_{\alpha}$ to be a quantifier of classical correlations are discussed and proved.  
From the main result, Renyi's entanglement of formation for 
pure states is calculated, recovering Ref.\cite{sanders2010}. 
The values of $C_{\alpha}$ 
for pure states and classical correlated states are also calculated, recovering the 
classical Jensen Shannon Divergence. As an application of the main result, it is shown 
that the entanglement of formation, for $\alpha = 1/2$, is a convex roof of 
log generalized robustness. This final result is discussed in the context of channel capacity in asymptotic limit. 

This paper is organized as follow.   In Section \ref{sec.math} some 
mathematical concepts about density matrix, the CPTP channels over density operators, 
the Schatten - p norm for operators, and entropic measure for quantum systems are  introduced. 
In Section \ref{sec.qc}  the formalism of quantum correlations: quantumness 
of correlations and quantum entanglement is presented. The main results are presented in Section \ref{sec.result}. 
In Section \ref{log.rob} the results are 
applied for $\alpha = 1/2$ Renyi's entropy, a relation between $C_{\alpha}$, 
the distinguishability of the states in the ensemble created by local measurement and the log robustness is obtained.

\section{Formalism and Notation}\label{sec.math}
\subsection{Density Operator}\label{density}
This work deals with finite dimensional Hilbert spaces. A given Hilbert space $\M{H}_N = \B{C}^{N}$ is 
defined as a complex vector space. 
For a given Hilbert space there exists a dual space $\M{H}^{*}_N$, that maps $\M{H}$ to the complex numbers. For finite dimensional Hilbert spaces 
these two spaces are isomorphic, so $\M{H}_N^{*} = \B{C}^{N}$. 
The space of linear transformations acting on  a 
Hilbert space is denoted as $\M{L}(\B{C}^N,\B{C}^M)$. 
A given linear transformation $A$ belongs to the space $\M{L}(\B{C}^N,\B{C}^M)$, 
if $A:\B{C}^N\rightarrow\B{C}^M$. If $A$ is an operator, its space is denote as $\M{L}(\B{C}^N)$. 
The set 
of linear transformations on the Hilbert space $\M{L}(\B{C}^N,\B{C}^M)$ 
is also a Hilbert space, therefore it is equipped with inner product. 
For two operators $M,N \in \M{L}(\B{C}^N)$, the inner product is defined as the Hermitian form:
\begin{equation}
\langle M,N\rangle = \Tr(M^{\dagger}N). 
\end{equation}
As $\Tr(M^{\dagger}N)$ is a finite number, the vector space $\M{L}$ is 
often called the space of {\it bound operators}. 
Restricting the matrices in the positive cone of this space to be trace=$1$, it defines  
the set of density matrices   \cite{zyc}. This set of operators is a vector space denoted as $\M{D}(\B{C}^N)$. 
\begin{defi}
A linear positive operator $\rho\in\M{D}(\B{C}^N)$ is a density matrix, and represents the state of 
a quantum system, if it satisfies the following properties:
\begin{itemize}
\item{Positive semi-definite}:  $ \rho \geq 0;$ 
\item{Trace one}:  $\Tr(\rho) =1$.
\end{itemize} 
\end{defi}

A linear transformation $\Phi: \M{L}(\B{C}^N)\rightarrow \M{L}(\B{C}^M)$ represents a physical process if it is completely positive and preserves the trace. In 
other words, the transformation $\Phi$ satisfies: 
\begin{itemize}
\item{\bf Completely Positive:} Consider a composed system described by $\sigma_{AB}\in\M{D}(\B{C}_{A}\otimes\B{C}_{B})$
\begin{equation}
\id_A\otimes\Phi_B(\sigma_{AB})\geq 0.
\end{equation}
\item{\bf Trace preserving:} For a given density matrix $\rho\in\M{D}(\B{C}^N)$
\[ \Tr[\Phi(\rho)] = \Tr[\rho] =1 \]
\end{itemize}
 Satisfying these properties the transformations map a density matrix into another density matrix, named a Completely Positive and Trace Preserving (CPTP)  channel.  
The set of CPTP channels is  denoted as $\M{P}(\B{C}^N,\B{C}^M)$.

Isometric transformations are 
linear transformations $V:\B{C}^{N}\rightarrow\B{C}^M$ 
satisfying:
\begin{equation}
V^{\dagger}V = \id_{N},
\end{equation}
$\M{U}(\B{C}^{N},\B{C}^{M})$ is the set of isometric operations. 
Isometries that map the space $\B{C}^{N}$ on itself are
named {\it unitary operators}. The set of unitary operations on $\B{C}^N$ is denoted as $U\in\M{U}(\B{C}^N)$. 
A unitary operator $U\in\M{U}(\B{C}^{N})$, satisfies $U^{\dagger}U=UU^{\dagger}=\id_{N}$. 
An isometric transformation preserves the inner product, and consequently the spectra of 
the operators. 

An important CPTP channel for this work is the local measurement map 
$M_A\in\M{P}(\B{C}_A,\B{C}_X)$. Where $\text{dim}(\B{C}_A) = |A|$ and  $\text{dim}(\B{C}_X) = |X|$. For projective measurements, the measurement map is 
simply the dephasing operation in the measured orthonormal basis. Thus, considering a 
bipartite state $\rhoAB$ and a local projective measurement over $A$, the post-measurement state 
is: 
\begin{equation}
M_A\otimes\id_B(\rho_{AB}) = \sum_x p_x \ketbra{a_x}{a_x}\otimes\rho_x^B, 
\end{equation}
where $\{\ket{a_x}\}_{1}^{|A|}$ is an orthonormal basis, and 
$p_x\rho_x^B = \Tr_A(\ketbra{a_x}{a_x}\otimes\id_B \rho_{AB})$. 
The state  $\rho_x^B$ is an output of the measurement process with probability $p_x$.
For general measurement processes described  by positive-operator valued measure  
(POVM), by Naimark\apos s dilation theorem the approach is the same as for projective measurements, the only difference 
is in the cardinality of the random variable $X = \{p_x \}_{x=1}^{|X|}$, that in this case, is equal to 
the number of POVM elements \cite{naimark}. The local measurement outputs create an ensemble of quantum states 
$\xi_{AB}^{M} = \{p_x,\rho_x^B \}_{x=1}^{|X|}$. 
 
\subsection{Schatten-p norm}\label{snorm}

The Schatten-p norm for operators is analogous to the $l_p$ norm for 
vectors. 
\begin{defi}
Given a linear operator $A\in\M{L}(\B{C}^N,\B{C}^M)$, the Schatten-p norm is defined as:
\begin{equation}
\Vert A \Vert_p = \{\Tr[(AA^{\dagger})^{p/2}]\}^{1/p}, 
\end{equation}
where $p = [1,\infty)$.
\end{defi}
The Schatten-p norm can be written as the $l_p$ norm of the spectral 
decomposition of the matrix $A$:
\begin{equation}
\Vert A \Vert_p = \left\{\sum_k |\lambda(A)_k|^{p}]\right\}^{1/p}, 
\end{equation}
where $\{\lambda(A)_k\}_k$ are the eigenvalues of $A$.
As the Schatten norm only depends on the eigenvalues of the matrix, it is invariant under action of isometries.

\subsection{Quantum Entropies}
Considering a random variable $X = \{p_x\}$, where $p_x\geq 0$ and $\sum_x p_x =1$, 
its $\alpha$ - Renyi\apos s entropy  is defined as \cite{renyi61}: 
\begin{equation}
H_{\alpha}(X) = \frac{1}{1-\alpha} \log\sum_x p^{\alpha}_x.
\end{equation}
For $\alpha\rightarrow 1$ it is the Shannon entropy of $X$: 
\[H_{1}(X) =- \sum_x  p_x \log p_x. \] 
The quantum version of the Renyi\apos s entropy is defined as \cite{wehrl78}: 
\begin{defi}[Quantum $\alpha$ entropy] 
Given a quantum state $\rho\in\M{D}(\B{C}^N)$, its $\alpha$ entropy is 
\begin{equation}
S_{\alpha}(\rho) = \frac{1}{1-\alpha} \log\N{\rho}_{\alpha}^{\alpha},
\end{equation}
where $\N{\rho}_{\alpha}$ is the Schatten  norm, then  
$\N{\rho}_{\alpha}^{\alpha} = \Tr(\rho^{\alpha})$.
\end{defi}
For quantum Renyi\apos s  entropy, the following proposition is introduced. 
\begin{prop}\label{propconc}
Renyi\apos s entropy is a concave function for $\alpha \in (0,1)$: 
\begin{equation}
S_{\alpha}(\sum_k p_k \rho_k) \geq \sum_k p_k S_{\alpha}(\rho_k)
\end{equation}
\end{prop}
The proof of this proposition was performed in Ref.\cite{raviv78}, and a modern discussion about this issue can be found on \cite{mosonyi2011}.
For $\alpha \rightarrow 1$ the $\alpha$ - entropy is equal the von Neumann entropy  \cite{wehrl78}: 
\begin{equation}
\lim_{\alpha\rightarrow 1} S_{\alpha}(\rho) = S(\rho) = -\Tr\left(\rho\log\rho  \right). 
\end{equation}

For composed systems, the total amount of correlations in a bipartite state 
$\rhoAB$ can be quantified by means of mutual information:
\begin{equation}
I(A:B)_{\rho_{AB}} = S(\rho_A) + S(\rho_B) - S(\rho_{AB}),  
\end{equation}
where $\rho_A = \Tr_B(\rho_{AB})$, the same for $\rho_B$.
Mutual information is zero if the state is a product state $\rho_{AB} = \rho_A\otimes\rho_B$. 
For classical quantum states $\rho_{AX} = \sum_x p_x \rho_x^A\otimes\ketbra{x}{x}$, 
where $X=\{p_x\}$ is a classical register, the mutual information is: 
\begin{equation}
I(A:X)_{\rho_{AX}} = S(\rho_A) - \sum_x p_x S({\rho_x^A}). 
\end{equation}
This function quantifies the distinguishability of states in the ensemble 
$\xi_A = \{p_x,\rho_x^A \}$, and is also named Jensen Shannon divergence \cite{majteyjensen}, 
represented in this work as
\begin{equation}
Q(\xi_A) = S(\rho_A) - \sum_x p_x S({\rho_x^A}). 
\end{equation}
For an ensemble of quantum state composed of two states, 
it is defined as the symmetric and smoothed version of 
Shannon relative entropy \cite{majteyjensen,jensenentropy}.  
It is related to the Bures distance and 
induces a metric for pure quantum states, related to the 
Fisher-Rao metric \cite{jensenmetric}. Holevo\apos s Theorem states that quantity is the 
capacity of a given channel to transmit classical information \cite{holevo1973information,schumacher97}.

\section{Quantum Correlations}\label{sec.qc}

This section presents some definitions and discussions about quantum correlations. 
First the class of classical correlated state 
and the definition of a quantifier of classical correlations are presented. 
Next  the concept of quantum entanglement is discussed. 
Finally  
the Koashi - Winter relation is introduced, which interplays classical correlations and quantum entanglement. 
\subsection{Quantumness of Correlations}
A composed state $\rho_{AB}\in\M{D}(\B{C}_A\otimes\B{C}_B)$ is said to be classically correlated if 
a local projective measurement $\Pi_A\otimes\Pi_B$, that commutes with the state, exists \cite{opp,piani08,zurek2001,henderson}
\begin{equation}
[\rho_{AB},\Pi_A\otimes\Pi_B] =0. 
\end{equation}
The class of states that satisfies this equation is composed of states in the following form:
\begin{equation}\label{classcorsta}
\rho_{AB} = \sum_{i=1}^{|A|}\sum_{j=1}^{|B|}p_{i,j} \Pi^A_i\otimes\Pi^B_j, 
\end{equation}
where $\sum_{i=1}^{|A|}\sum_{j=1}^{|B|}p_{i,j}=1$, $p_{i,j}\geq 0$ and $\Pi_A\otimes\Pi_B\in \{\Pi^A_i\otimes\Pi^B_j\}_{i,j}$. 
The amount of classical correlations in a quantum state is measured by the capacity 
to extract information locally \cite{fanchini2012}. As the measurement process is a classical 
statistical inference, classical correlations can be quantified  by the amount of correlations 
remaining in the system after a local measurement \cite{piani08}. 
\begin{defi}[Classical Correlations]
For a bipartite density matrix $\rho_{AB}\in\M{D}(\B{C}_A\otimes\B{C}_B)$, classical 
correlations between $A$ and $B$ can be quantified by the amount of correlations 
 extracted by means of local measurements:
\begin{equation}\label{f1}
J(A:B)_{\rho_{AB}}= \max_{\id\otimes \M{B}\in\M{P}}{I(A:X)_{\id\otimes \M{B}(\rho_{AB})}},  
\end{equation}
where the optimization is taken over the set of local measurement maps 
$\id\otimes \M{B}\in\M{P}(\B{C}_{AB},\B{C}_{AX})$ and  
$\id\otimes \M{B}(\rho_{AB}) = \sum_x p_x \rho_{x}^{A}\otimes\ketbra{b_x}{b_x}$ 
is a quantum-classical state in the space $\M{D}(\B{C}_A\otimes\B{C}_X)$. 
\end{defi}

As the mutual information quantifies the total amount of 
correlations in the state, it is possible to define a quantifier of quantum 
correlations as  
the difference between the total correlations in the system, 
quantified by mutual information, and  the classical correlations, 
measured by Eq.\eqref{f1}. This measure of quantumness of correlations 
is named {\it quantum discord} \cite{zurek2001,henderson}:
\begin{equation}\label{f2}
D(A:B)_{\rho_{AB}} = I(A:B)_{\rho_{AB}} - J(A:B)_{\rho_{AB}},
\end{equation}
where $I(A:B)_{\rho_{AB}}$ is the von Neumann mutual information. 
The quantum discord quantifies the amount of information,  
that cannot be accessed via local measurements \cite{fanchini2012}, therefore 
it measures the quantumness of correlations between $A$ and $B$ 
that cannot be recovered via a classical statistical inference process. 
\subsection{Entanglement}
A pure state $\ket{\psi}_{AB}\in\B{C}_A\otimes\B{C}_B$ is uncorrelated if 
it can be written as a tensor product of pure states of each partition:
\begin{equation}
\ket{\psi}_{AB} = \ket{a}\otimes\ket{b},
\end{equation}
where $\ket{x}\in\B{C}_X$, for $x=\{a,b\}$. 
By the definition of classical correlations in Eq.\eqref{classcorsta}, 
a convex combination of product states can be quantum correlated. Therefore,  taking convex combinations of non orthogonal states results in the 
notion of  {\it separable state} \cite{werner89}. 
\begin{defi}[Separable states]\label{defsep}
Considering a composed system described by the state\\ 
$\sigma\in\M{D}(\B{C}_{A}\otimes\B{C}_{B})$, it is a 
{\it separable state} if and only if it can be written as:
\begin{equation}\label{sep}
\sigma_{AB} = \sum_{i,j}p_{i,j}\ketbra{\psi_i}{\psi_i}_A\otimes\ketbra{\phi_i}{\phi_i}_B,
\end{equation} 
where $\ket{\psi_i}_A\in\B{C}_A$ and $\ket{\phi_i}_B\in\B{C}_B$.
\end{defi}
Note that the states $\{\ket{\psi_i}_A\}_{i}$ and $\{\ket{\phi_j}_B\}_{j}$ are, in general, not orthogonal states. If these sets are composed of orthogonal states, the state in Eq.\eqref{sep} is 
classically correlated.  Quantum entanglement is defined as the negation of Definition \ref{defsep} \cite{schrodinger35}:
\begin{defi}[Entanglement]
A composed state $\rhoAB$ is entangled if it is not separable.
\end{defi} 
The amount of quantum entanglement of a bipartite system $\rhoAB$ can be quantified by the entanglement of formation. 
The entanglement 
of formation is interpreted as the minimum amount of entangled 
pure states required to build $\rho_{AB}$, by means of a convex combination  \cite{bennettpra}.
\begin{defi}[Entanglement of Formation]
Considering a quantum state $\rho\in\M{D}(\B{C}_{A}\otimes\B{C}_{B})$, 
the entanglement of formation is defined as the convex roof:
\begin{equation}
E_f(\rho) = \inf_{\xi_{\rho}}\sum_i p_i E(\ket{\psi_i}),
\end{equation}
where the minimization is performed over all ensembles $\xi_{\rho}=\{p_i,\ketbra{\psi_i}{\psi_i}\}_{i=1}^M$, 
such that $\rho = \sum_i p_i \ketbra{\psi_i}{\psi_i}$, $\sum_i p_i =1$ and $p_i\geq 0$. 
\end{defi}
The function $E(\ket{\psi_i}) = S(\rho_A)$ is named entropy of entanglement \cite{schumacher95,ekert95}, and is the usual quantifier of entanglement 
for pure states.  
For pure states, entanglement of formation is equal to the classical correlations and quantum discord \cite{fanchini11}:
\begin{equation*}
E_f({\psi}_{AB}) = J(A:B)_{\psi_{AB}} = D(A:B)_{\psi_{AB}} = S(\rho_A), 
\end{equation*}
where $\ketbra{\psi}{\psi}_{AB}\in\M{D}(\B{C}_A\otimes\B{C}_B)$ is a pure state and 
$\rho_A = \Tr_B({\psi})$. 

As in this work the  interest is in Renyi\textsc{\char13}s entropies, the $\alpha$ - Entanglement of Formation (EoF) is introduced \cite{sanders2010}: 
\begin{defi}
$\alpha$ - entanglement of formation of a bipartite state $\rhoAB$ is defined as:
\begin{equation}
E_f^{\alpha}(\rho_{AB}) =  \inf_{\{p_i,\ketbra{\psi_i}{\psi_i}\}_i}\sum_i p_i S_{\alpha}(\Tr_B(\ketbra{\psi_i}{\psi_i}_{AB})
\end{equation}
for $\rho_{AB} = \sum_i p_i \ketbra{\psi_i}{\psi_i}_{AB}$.
\end{defi}
$S_{\alpha}(\Tr_B(\ketbra{\psi_i}{\psi_i}_{AB})$ is the $\alpha$ - entropy of entanglement, which  is an entanglement monotone for $\alpha\in (0,1)$ \cite{horodecki96,vidal99}.
The Schur concavity of $\alpha$ - entropy for $\alpha\in (0,1)$ guarantees that $\alpha$ - entanglement of formation is a monotone function under the action of local operations and classical communication (LOCC) \cite{sanders2010}.

\subsection{Koashi - Winter relation}
\label{seckw}

Given a bipartite system $\rho_{AB}\in\mathcal{D}(\mathbb{C}_A\otimes \mathbb{C}_B)$, 
and its purification $\ket{\psi}_{ABE}\in\B{C}_A\otimes\B{C}_B\otimes\B{C}_E$. 
The dimension of the global space is: 
$\text{dim}(\mathbb{C}_{ABE})=\text{dim}(A)\cdot\text{dim}(B)\cdot \text{rank}(\rho_{AB})$. 
The purification creates quantum correlations between the 
system $AB$ and the purification system $E$, unless the state is already 
pure. The balance between the correlations of a tripartite purification is settled by the 
 Koashi-Winter relation \cite{kw04}. 

\begin{teo}[Koashi-Winter relation]
Considering $\rho_{ABE}\in\mathcal{D}(\mathbb{C}_A\otimes \mathbb{C}_B \otimes \mathbb{C}_E)$ a pure state then:
\begin{equation}\label{kw}
J(A:E)_{\rho_{AE}} = S(\rho_{A}) - E_f(\rho_{AB}),
\end{equation}
where $\rho_X = Tr_{Y}[\rho_{YX}]$.
\end{teo}
The Koashi-Winter equation quantifies the amount of entanglement among $A$ and  
$B$, considering that the former is classically correlated with another system $E$. 
This property 
is interesting as it is related with the monogamy of entanglement 
\cite{wooters00}. 
An analogous expression of the Koashi - Winter relation has been obtained for quantum discord and entanglement of formation \cite{fanchini11}. 
From this new relation, the irreversibility of the entanglement distillation protocol  
and quantum discord  are interplayed \cite{cornelio11}. This result can be applied for the state merging protocol \cite{esm,dattasm}.   
In addition to the above relation, some upper and lower bounds between quantum discord 
and entanglement of formation have been calculated via the Koashi-Winter relation 
and entropic properties  \cite{yu,xi,xib,zhang}.
By means of Eq.\eqref{kw}  quantum discord and 
entanglement of formation were calculated analytically for systems with  
rank-$2$ and dimension 
$2\otimes n$ \cite{fanchini2012,lastra,cen}.

In this work a generalization of the Koashi - Winter relation is calculated  for a class of Renyi\apos s entropies for the parameter $\alpha \in (0,1)$.  

\section{Results}\label{sec.result}
First the $\alpha$ Quantum Jensen Shannon divergence (QJSD) is introduced  \cite{briet2009}.
\begin{defi}[$\alpha$ - Quantum Jensen Shannon divergence]
Given a quantum ensemble $\xi = \{p_k, \rho_k \}_{k}^{M}$, for $\rho_k\in\M{D}(\B{C}^N)$ the Renyi\apos s quantum Jensen Shannon divergence, for $\alpha \in (0,1)$ 
is defined as: 
\begin{equation}
Q_{\alpha}(\xi) = S_{\alpha}(\sum_k p_k \rho_k) - \sum_k p_k S_{\alpha}(\rho_k),  
\end{equation}
where \[S_{\alpha}(X) = \frac{1}{1-\alpha}\log\norm{X}^{\alpha}_{\alpha}\] and $\norm{\cdot}_\alpha$ is the Schatten norm \[ \norm{X}_{\alpha}^{\alpha} = \Tr(X^{\alpha}),\] for any 
matrix $X\in\M{L}(\B{C}^N)$. 
\end{defi}
A corollary of the concavity of the Renyi\apos s entropy is the positivity of the $\alpha$ - QJSD function: 
\begin{equation}
Q_{\alpha}(\xi) \geq 0,
\end{equation}
for $\alpha \in (0,1)$. The $\alpha$ - QJSD is zero if the ensemble has cardinality equal to one ($M = 1$), and maximum if the states in the ensemble are 
pure and lineally independent. The $\alpha$ - QJSD is a generalization of the QJSD \cite{briet2009}, and quantifies the distinguishability between the states in the ensemble. 

As aforementioned, local measurements create a quantum ensemble in the non measured partition, composed by the output states. Given this property  the function is defined: 
\begin{defi}
For a bipartite state $\rho_{AB}$, one can define the function 
\[C_{\alpha}(\rho_{AB}) = \sup_{\xi_{AB}^{M_B}} Q_{\alpha}(\xi_{AB}^{M_B}), \]
where  $\xi_{AB}^{M_B} = \{ p_x, \rho_x^{A} \}$ for $p_x\rho_x^A = \Tr_B(\id_A\otimes M_x \rho_{AB})$ and $\{ M_x\}_x$ are elements of a POVM.  
\end{defi}
As discussed $Q_{\alpha}(\xi_{AB}^{M_B})$ quantifies the distinguishability between the states in the ensemble $\xi_{AB}^{M_B}$, therefore $C_{\alpha}(\rho_{AB})$ 
measures the distinguishability between the states in the ensemble created by the measurement $M_B$ such that the states are the most distinguishable. For von Neumann entropy this quantity quantifies classical correlations of the state $\rho_{AB}$.  

The relation of $C_{\alpha}(\rho_{AB})$ with quantum correlations in $\rho_{AB}$ is stated in the main result of this work, presented below.

\begin{teo}\label{main}
Considering a pure tripartite state $\psi_{ABE} = \rho_{ABE} \in \M{D}(\B{C}_A\otimes \B{C}_B\otimes \B{C}_E)$, the following equality holds: 
\begin{equation}\label{t1}
C_{\alpha}(\rho_{AE}) = \sup_{\xi_{AE}^{\M{M}_E}} Q_{\alpha}(\xi_{AE}^{\M{M}_E}) = S_{\alpha}(\rho_A) - E_{f}^{\alpha}(\rho_{AB})
\end{equation}
\end{teo}
\begin{proof}
There exists a set of orthogonal states $\{\ket{l}\}_{l=1}^{\abs{E}}$ such that the states $\psi_{ABE}$ can be written as:
\[ \ket{\psi}_{ABE} = \sum_l p_l \ket{\phi_l}_{AB}\ket{l}_E, \]
thus the reduced density matrix $\Tr_E(\psi_{ABE})=\rho_{AB} = \sum_{l} p_l \ketbra{\phi_l}{\phi_l}$. Performing a measurement on subsystem $E$, such that the POVM elements 
of the measurement are rank-1 operators $\{ M_x^E = \ketbra{\mu_x}{\mu_x} \}$, where $\sum_x M_x^E = 1$ and $M_x^E \geq 0$, the post-measurement state is: 
\begin{align}
\rho_{ABE'} &= \id_{AB}\otimes\M{M}_E(\rho_{ABE})\\ &= \sum_x \Tr_E\left( \id_{AB}\otimes \ketbra{\mu_x}{\mu_x}_E \rho_{ABE} \right) \otimes \ketbra{e_x}{e_x}_{E'},\\
&=\sum_x q_x \ketbra{\psi_x}{\psi_x}_{AB}\otimes\ketbra{e_x}{e_x}_{E'},
\end{align}
where $q_x \ketbra{\psi_x}{\psi_x}_{AB} = \Tr_E\left( \id_{AB}\otimes \ketbra{\mu_x}{\mu_x}_E \rho_{ABE} \right)$. As $\rho_{AB} = \Tr(\rho_{ABE}) = \Tr(\rho_{ABE'})$, 
it is clear that there exists a POVM such that: 
\begin{align*}
q_x &= p_x \\
\ket{\psi_x}_{AB} &= \ket{\phi_x}_{AB}.
\end{align*}
In subsystem $AE'$ the post-measurement state is: 
\begin{align}
\rho_{AE'} &= \Tr_B\left( \rho_{ABE'}  \right)\\ &= \sum_x q_x \Tr_B(\ketbra{\psi_x}{\psi_x}_{AB}\otimes\ketbra{e_x}{e_x}_{E'}) \\
& = \sum_x q_x \rho_x^A\otimes\ketbra{e_x}{e_x}_{E'},
\end{align} 
where $\Tr_B(\ketbra{\psi_x}{\psi_x}_{AB}) = \rho_x^A$. The state $\rho_{AE'}$ represents the ensemble of quantum state $\xi_{AE}^{M_E} = \{q_x, \rho_x^A \}$ 
prepared according to the random variable $X = \{ q_x\}_x$. 
In this way calculating $\alpha$-QJSD for the ensemble $\xi_{AE}^{M_E}$:
\begin{equation}
Q_{\alpha}(\xi_{AE}^{\M{M}_E}) = S_{\alpha}(\rho_A) - \sum_x q_x S_{\alpha}(\Tr_B(\ketbra{\psi_x}{\psi_x}_{AB}). 
\end{equation}
The quantum ensemble $\xi_{AE}^{M_E}$ is created by means of a measurement $\M{M}_E$ on the subsystem $E$, implying that it is possible to find a measurement  
such that the created ensemble maximizes the $\alpha$-QJSD: 
\begin{equation}\nonumber
\sup_{\xi_{AE}^{M_E}} Q_{\alpha}(\xi_{AE}^{\M{M}_E})= S_{\alpha}(\rho_A) - \inf_{\xi_{AE}^{M_E}}\sum_x q_x S_{\alpha}(\Tr_B(\ketbra{\psi_x}{\psi_x}_{AB}).
\end{equation}
However the ensemble is created by means of a measurement performed on $E$, thus one can rewrite the last term of the equation as
\begin{align}
&\inf_{\xi_{AE}^{M_E}}\sum_x q_x S_{\alpha}(\Tr_B(\ketbra{\psi_x}{\psi_x}_{AB})\\ &= \inf_{\{q_x,\ketbra{\psi_x}{\psi_x}\}_x}\sum_x q_x S_{\alpha}(\Tr_B(\ketbra{\psi_x}{\psi_x}_{AB}).
\end{align}
As the state in $AB$, on average, does not change by the measurement on $E$, one can identify the right hand side of the equation as the 
$\alpha$-Renyi Entanglement of Formation: 
\[ E_f^{\alpha}(\rho_{AB}) = \inf_{\{q_x,\ketbra{\psi_x}{\psi_x}\}_x}\sum_x q_x S_{\alpha}(\Tr_B(\ketbra{\psi_x}{\psi_x}_{AB}),\]
for $\rho_{AB} = \sum_x q_x \ketbra{\psi_x}{\psi_x}_{AB}$.
\end{proof}
From Eq.\eqref{t1} it is possible to recover and generalize that, for pure states the entanglement of formation is equal to the 
von Neumann entropy of the reduced density matrix \cite{sanders2010}.
\begin{teo}[Pure States]\label{prop1}
Consider $\rho_{AB}$ a pure state, the $\alpha$ - QJSD of the ensemble $\xi_{AE}^{\M{M}_E}$ is zero, therefore: 
\begin{equation}\label{ps}
E_f^{\alpha}(\rho_{AB}) = S_{\alpha}(\rho_A).
\end{equation}
\end{teo}
\begin{proof}
The state $\rho_{AB}$ can be written in the Schmidt decomposition: 
\begin{equation}\label{sd}
\ket{\psi_{AB}} = \sum_l c_l \ket{a_l}\ket{b_l}. 
\end{equation}
The purification of a pure state is just coupling  another pure ancilla to it: 
\begin{equation}
\rho_{ABE} = \ketbra{\psi_{AB}}{\psi_{AB}}\otimes\ketbra{0}{0}. 
\end{equation}
Performing a measurement $\M{M}_E$ with POVM elements $\{ M_x^E\}_x$ on system $E$, the post-measurement states on subsystem $A$ are: 
\[\rho_x^A = \frac{1}{p_x} \Tr_E\left(\id_A\otimes M_x^E \rho_{AE} \right) = \frac{1}{p_x}  \sum_l c_l \bra{0} M_x \ket{0} \ketbra{a_l}{a_l},\] 
taking the partial trace over $B$ on Eq.\eqref{sd} one realizes that: 
\[  \rho_x^A = \frac{1}{p_x} \bra{0} M_x \ket{0} \rho_{A} = \rho_A ,\]
for every measurement map performed on $E$. 
Therefore the ensemble created by the local measurement is composed of only one single state $\xi_{AE}^{M_E} = \{1, \rho_A\}$, implying that Renyi\apos s QJSD of the ensemble is zero. 
\end{proof}

The Renyi\apos s entropy generalization of the Koashi - Winter relation state that there is 
an interplay between the most distinguishable states 
of the ensemble created by the local measurement 
on the bipartite system, and the $\alpha$ EoF of the unmeasured system and the purification ancillary system. 
The standard KW relation 
relates classical correlations and the entanglement of formation of the unmeasured state and 
the purification ancillary system. 
Note that the function $C_{\alpha}(\rho_{AE})=\sup_{\xi_{AE}^{M_E}} Q_{\alpha}(\xi_{AE}^{\M{M}_E})$ can be a quantifier of 
classical correlations.  
As discussed by Henderson and Vedral \cite{henderson}, the standard measure of classical correlations quantifies the amount 
of information accessed via local 
measurements on a bipartite system, that is equal to the distinguishability of the states of the ensemble created by
the local measurement, quantified by the QJSD. 
In this way, the properties that $C_{\alpha}$ must satisfies to be a measure of classical correlations are now discussed, and some analytical results are obtained from 
this discussion.

It is possible to rewrite Eq.\eqref{t1} changing the order of the 
labels $B\rightarrow E$:
\begin{equation}
C_{\alpha}(\rho_{AB}) = S_{\alpha}(\rho_A) - E_{f}^{\alpha}(\rho_{AE}). 
\end{equation}
Then now the properties of the function $C_{\alpha}(\rho_{AB})$ are presented, for a given density operator $\rho_{AB} \in \M{D}(\B{C}_A\otimes\B{C}_B)$, 
and state that it can be a quantifier of classical correlations of $\alpha \in (0,1)$.

For a function of information to quantify correlations between quantum systems it must satisfy some important properties \cite{brodutch12}. 
\begin{enumerate}
\item\label{p1} $C_{\alpha}(\rho_{AB}) =0$ if and only if $\rho_{AB}$ is a product state;
\item\label{p2} $C_{\alpha}(\rho_{AB})=C_{\alpha}(U_A\otimes U_B\rho_{AB}U_A^{\dagger}\otimes U_B^{\dagger})$, for $U_X\in\M{U}(\B{C}_X)$  
a unitary operation.
\item\label{p3} $C_{\alpha}(\rho_{AB})\geq C_{\alpha}({\Phi}_A\otimes{\Phi}_B(\rho_{AB}))$, for ${\Phi}_X$ a CPTP map. 
\end{enumerate}
The proof that $C_{\alpha}(\rho_{AB})$ satisfies these properties is performed in the sequence by the following theorems. 

\begin{teo}[Property \ref{p1}]
\label{tcc}
Consider a state $\rho_{AB}$ and its post local measurement state $\rho_{AB'} = \sum_{x}p_x \rho_{A} \otimes\ketbra{x}{x}_{B'}$, for the ensemble 
$\xi_{AB}^{M_B} = \{ p_x, \rho_x^A \}$ the Renyi QJSD, maximized over all ensembles created by the local measurement, is zero if and only if $\rho_{AB}$  
is a product state.
\end{teo}
\begin{proof}
Given $\rho_{AB}=\rho_A\otimes\rho_B$ its purified state will also be a product state:
\begin{align*}
\ket{\psi}_{ABER} &= \ket{\phi}_{AE}\otimes\ket{\varphi}_{BR}\\ 
&= \left(\sum_{l}\sqrt{a_l}\ket{a_l}_A\ket{l}_E\right)\otimes\left(\sum_k\sqrt{k}\ket{b_k}_B\ket{k}_R\right),
\end{align*} 
for $\rho_{A}=\sum_{l}{a_l}\ketbra{a_l}{a_l}$ and $\rho_{B}=\sum_{k}{b_k}\ketbra{b_k}{b_k}$. As shown in Proposition \ref{prop1}, $\alpha$ - entanglement of formation for 
pure state is equal to the Renyi\apos s entropy of the reduced density matrix: 
\begin{equation}
E_f^{\alpha} (\ket{\phi}_{AE}) = S_{\alpha}(\rho_A),
\end{equation}
therefore the $\alpha$-QJSD is zero. 
On the other hand, the ensemble $\xi_{AB}^{M_B}$ is created by means of a local measurement $\M{M}$ over $B$ on the product state $\rho_{AB} = \rho_A\otimes\rho_B$:  
\[ \rho_{AB'} = \rho_A\otimes\M{M}(\rho_B). \] 
For every measurement performed over the subsystem $B$ the state in $A$ remains  undisturbed:
\[ \rho_A\otimes\M{M}(\rho_B)=\rho_A\otimes\rho_{B'}.  \]
Therefore the ensemble created in $A$ by means of this measurement over $B$ has just one element $\xi_{AB}^{M_B} = \{ 1, \rho_A \}$, which implies:
\begin{equation}
Q_{\alpha}(\xi_{AB}^{M_B}) = 0.
\end{equation}
\end{proof}

\begin{teo}[Property \ref{p2}]\label{t.2}
$C_{\alpha}(\rho_{AB})$ is invariant under local unitary operations. 
\end{teo}
\begin{proof}
As Schatten-$p$ norm is invariant under unitary operations, then $\alpha$- QJSD is  invariant under unitary operations: 
\begin{align*}
Q(U \xi U^{\dagger})_{\alpha} &= S_{\alpha}(U\left(
\sum_k p_k \rho_k\right)U^{\dagger}) - \sum_k p_k S_{\alpha}(U\rho_kU^{\dagger})\\ 
& = Q( \xi)_{\alpha},
\end{align*}
where $U \in \M{U}(\B{C}_{\Gamma})$ is a unitary operation. 
As it is holds for every ensemble $\xi = \{p_k,\rho_k \}_{k=1}^{M}$.
\end{proof}

\begin{teo}[Property \ref{p3}]\label{t.3}
For a bipartite state $\rho_{AB}$, the function 
\[C_{\alpha}(\rho_{AB}) \geq  C_{\alpha}(\tilde{\rho}_{AB}), \]
for $\alpha\in(0,1),$ 
where $\tilde{\rho}_{AB} = \Phi_A\otimes\Phi_B(\rho_{AB})$, and 
$\Phi_X\in\M{P}(\B{C}_X)$. 
\end{teo}
\begin{proof}
This comes from the fact that $\alpha$ - EoF is an entanglement monotone, 
and decreases under LOCC, for $\alpha\in(0,1)$.
\end{proof}

An interesting analytical result obtained from Eq.\eqref{t1} is  that $C_{\alpha}(\rho_{AB})$ 
is equal to the entropy of entanglement for $\rho_{AB}$ a pure state.
\begin{teo}[Classical correlations of a pure state]\label{ps.cc}
For $\rhoAB$ a pure state $\rho_{AB} = \ketbra{\psi}{\psi}_{AB}$, 
it holds 
\begin{equation}
C_{\alpha}(\rho_{AB}) = S(\rho_A),
\end{equation}
where $\rho_A = \Tr_B(\ketbra{\psi}{\psi}_{AB})$. 
\end{teo} 

\begin{proof}
As $\rho_{AB} = \ketbra{\psi}{\psi}_{AB}$ is a pure state, 
its purification $\ket{\phi}_{ABE} = \ket{\psi}_{AB}\otimes\ket{0}_E$ is a product state in the space 
$\B{C}_{AB}\otimes\B{C}_E$, then $E_{f}^{\alpha}(\rho_{AE})=0$,  
proving the statement. 
\end{proof}

For quantum states without quantum correlations it is expected that the amount of classical correlations is 
equal to a standard classical entropy. This is obtained in the next theorem. 

\begin{teo}[Classically Correlated State]\label{t.clascor}
Considering $\rho_{AB}$ a classical correlated state: 
\[ \rho_{AB} = \sum_{x,y} p_{x,y} \ketbra{a_x}{a_x}\otimes\ketbra{b_y}{b_y}, \]
where $\{\ket{a_x}\}_{x=1}^{|A|}$ and $\{\ket{b_y}\}_{y=1}^{|B|}$ are 
orthonormal basis in $\B{C}_A$ and $\B{C}_B$ respectively, then: 
\[ C_{\alpha}(\rho_{AB}) = H_{\alpha}(X,Y) - H_{\alpha}(X|Y), \]
for $X = \{ p_x = \sum_{y}p_{x,y}\}_{x=1}^{|A|}$, and analogous for $Y$. 
\end{teo}

\begin{proof}
Taking the local measurement over partition $B$, there exists a 
measurement operation $\M{M}_B$ that enables the classically correlated state invariant 
\[ \id_A \otimes \M{M}_B(\rho_{AB}) = \rho_{AB}.  \]
The post-measurement ensemble of states is:
\[ \xi_{AB}^{\M{M}_B}=\left\{p_y , \sum_x p(x|y)\ketbra{a_x}{a_x} \right\}_{y=1}^{|B|}, \] 
where $p(x|y) = p_{x,y}/p_y$. Calculating $\alpha$ - QJSD
of  $\xi_{AB}^{\M{M}_B}$: 
\begin{align} Q_{\alpha}(\xi_{AB}^{\M{M}_B})&=
S_{\alpha}(\sum_{x,y} p_{y}p(x|y) \ketbra{a_x}{a_x})- \\& - 
\sum_{y} p_{y} S_{\alpha}(\sum_x p(x|y) \ketbra{a_x}{a_x}), \end{align}
as: 
\begin{align}
S_{\alpha}(\sum_{x,y} p_{y}p(x|y) \ketbra{a_x}{a_x}) = H_{\alpha}(X,Y)\\
 S_{\alpha}(\sum_x p(x|y) \ketbra{a_x}{a_x}) = H_{\alpha}(X|Y),
\end{align}
proving the proposition. 
\end{proof}

\begin{remark} This definition for classical conditional entropy 
\[ H_{\alpha}(X|Y) = \frac{1}{1-\alpha} \sum_y p_y \log[\sum_x p(x|y)^{\alpha}], \]
does not satisfy the monotonicity under stochastic operations for every $\alpha \in (0,1)\cup(0,\infty)$ \cite{matos12,fehr2014}. 
An interesting discussion about this issue can be found in Ref.\cite{tomathesis}, although it is not known if this is not monotone for 
every $\alpha\in (0,1)$. 
\end{remark}

\section{KW - relation for Log Robustness}\label{log.rob}
As an application of the results of this paper, an 
interesting measure of entanglement is the well known {\it generalized robustness}, which  
quantifies the amount of mixture with another state 
needed to destroy the entanglement of the system \cite{vida99,steiner03}.   
Formally this is defined as:
\begin{defi}[Generalized robustness]
Consider an n-partite state $\rho\in\M{D}(\B{C}_{A_1}\otimes\cdots\otimes\B{C}_{A_n})$, 
 generalized robustness of $\rho$ is defined as:
\begin{equation}
R_G(\rho) = \left\{\min_{s\in\B{R}_{+}} s\, : \exists \rho_s\quad s.t.\quad \frac{\rho + s\rho_s}{1+s} \in \text{Sep}\right\},
\end{equation}
where $Sep$ is the set of separable states in $\M{D}(\B{C}_{A_1}\otimes\cdots\otimes\B{C}_{A_n})$.
\end{defi}  
The parameter $s$ is zero for separable states and finite for entangled states \cite{vida99}.

Another entanglement quantifier related to the generalized robustness of entanglement is the {\it log - generalized robustness} (LGR) defined as 
\cite{brandaopra2005}:
\begin{equation}\label{loggr}
LR_g(\rho) = \log_2(1+R_G(\rho)),
\end{equation}
where $R_G(\rho)$ is the generalized robustness of $\rho$. LGR is an entanglement monotone, sub-additive, non increasing 
under trace preserving separable operations, and an upper bound for the distillable entanglement \cite{brandaopra2005}.  It was also 
studied in the context of the resources theory of quantum entanglement \cite{brandao10,bnature,datta09}. 

As an application of the main results in 
Eq.\eqref{t1}, it is possible to obtain that for $\alpha = 1/2$, Renyi\apos s entanglement entropy  
of a pure state is  equal to the LGR for the pure state.

\begin{lemma}
Considering a pure state $\ket{\psi}_{AB}\in\B{C}_A\otimes\B{C}_B$, the 
$\alpha=1/2$ - entanglement entropy is equal to the log - robustness:   
\[ S_{1/2}(\rho_A) = LR_g(\psi_{AB}),\]
where $\rho_A = \Tr_B(\psi_{AB})$ and $\psi_{AB} = \ketbra{\psi}{\psi}_{AB}$. 
\end{lemma}
\begin{proof}
Considering the pure state in its Schmidt decomposition 
$\ket{\psi}_{AB} = \sum_i \sqrt{\mu_i}\ket{a_i}\ket{b_i}$, then its reduced density matrix is $\rho_A = \sum_i \mu_i \ketbra{a_i}{a_i}$. The $\alpha=1/2$ entropy is simply:
\[ S_{1/2}(\rho_A) = 2 \log\Tr(\sqrt{\rho_A}) = 2 \log(\sum_i \sqrt{\mu_i}). \]
As bipartite pure state the generalized robustness is \cite{vidal99}:
\[ R_G(\psi_{AB}) = \left(\sum_i \sqrt{\mu_i}\right)^2 -1, \]
then by definition of LGR:
\[ S_{1/2}(\rho_A) = LR_g(\psi_{AB}). \] 
\end{proof}
As a direct corollary, it is possible to calculate that the $\alpha=1/2$ - entanglement of 
formation is a convex roof version of the LGR: for a bipartite 
state $\rho_{AB}\in\M{D}(\B{C}_A\otimes\B{C}_B)$
\begin{equation}\label{eoflr}
E_{f}^{1/2}(\rho_{AB}) = \min_{\{p_i,\ket{\psi_i}\}}\sum_i p_iLR_g(\psi_i),
\end{equation}
where $\rho_{AB} = \sum_i p_i  \ketbra{\psi_i}{\psi_i}_{AB}$. 

Before introducing the main theorem of this section, some useful lemmas are proved.
\begin{lemma}
Given an ensemble of quantum states $\xi_A = \{ p_x, \rho_{x}^A\}_{x=1}^{|X|}$, for $ \rho_{x}^A\in\M{D}({C}_A)$ and $|X|$ the cardinality of classical distribution 
$X  = \{ p_x\}$. The probability of success in distinguishing the 
states in the ensemble is defined as:
\[ P_{suc}(X|A) = \sup_{\{E_x\}}\sum_x \Tr(E_x \rho_x^A), \]
where $\{E_x\}_{x=1}^{|X|}$ is a set of POVM elements. 
It is rated by the $\alpha=1/2$ entropy of $\rho_A = \sum_x p_x \rho_x^A$ as:
\begin{equation}\label{psuc}
S_{1/2}(\rho_A) \geq - \log{P_{suc}(X|A)}.
\end{equation}
\end{lemma}

\begin{proof}
As $\alpha$ entropy monotonically increases for $\alpha \in (0,1)\cup(1,2]$ \cite{petz86}, then:
\[ S_{1/2}(\rho_A) \geq S_{1/2}(\rho_{AX}),\]
where $\rho_{AX}= \sum_x p_x \rho_x^A\otimes\ketbra{x}{x}_X$. 
As discussed in Ref.\cite{tomathesis} $S_{1/2}$ is named $S_{max}$, or $max$ entropy. Therefore:
\begin{align}
 S_{1/2}(\rho_A) &\geq S_{max}(\rho_{AX})\\
\label{l.3.1.5.renner} & \geq S_{min}(X|A) \\
\label{t.1.renner} & = - \log{P_{suc}(X|A)} .
\end{align}
Where $S_{min}(X|A) = \max_{\sigma} \left\{ \N{\sigma_A^{-1/2} \otimes \id \rho_{AX} \sigma_A^{-1/2} \otimes \id}_2 \right\}$ \cite{renner}.  
Eq.\eqref{l.3.1.5.renner} is Lemma 3.1.5 of Ref.\cite{rennerthesis} and 
Eq.\eqref{t.1.renner} is Theorem.1 of Ref.\cite{renner}. 
\end{proof}

\begin{lemma}
Consider $\rhoAB$, the regularized $E_{1/2}^{\infty}(\rho_{AB})$ and $ LR_g^{\infty}(\rho_{AB})$ defined respectively as:
\[  E_{1/2}^{\infty}(\rho_{AB}) = \lim_{n\rightarrow \infty} \frac{E_f^{1/2}(\rho_{AB}^{\otimes n})}{n},\]
\[  LR_g^{\infty}(\rho_{AB}) = \lim_{n\rightarrow \infty} \frac{LR_g(\rho_{AB}^{\otimes n})}{n},\]
 the following equality holds:
\begin{equation}\label{lrgeof} E_{1/2}^{\infty}(\rho_{AB}) =  LR_g^{\infty}(\rho_{AB}).\end{equation}
\end{lemma}
\begin{proof}
Consider the relative entropy of entanglement:
\[ E_R(\rho) = \min_{\sigma\in\text{Sep}} S(\rho||\sigma), \]
where $S(\rho||\sigma) = -\Tr(\sigma\log\rho) - S(\rho)$ is the relative entropy. For pure states $E_R(\psi)= S(\rho_A)$, where $\rho_A$ is 
the reduced density matrix of $\psi = \ketbra{\psi}{\psi}$. 
As demonstrated by Brand\~ao and Plenio \cite{bnature,brandao10}: $ LR_g^{\infty}(\rho_{AB}) = E_R(\rho_{AB}) = E_C(\rho_{AB})$, where $E_C$ is the entanglement cost. 
It implies that for pure states: 
\[  LR_g^{\infty}(\psi_{AB}) =  E_R(\psi_{AB}) = S(\rho_A),\]
which implies that $E^{\infty}_{f}(\rho_{AB}) = E^{\infty}_{1/2}(\rho_{AB})$.
As proved in Ref.\cite{hayden01}: $E_{f}^{\infty}(\rho_{AB}) = E_C({\rho_{AB}})$, where $E_{f}^{\infty}$ is the regularized entanglement of formation. 
Therefore the statement comes from:
\[ LR_g^{\infty}(\rho_{AB})=E_C(\rho_{AB}) = E^{\infty}_{f}(\rho_{AB}) = E^{\infty}_{1/2}(\rho_{AB}). \]
\end{proof}

As aforementioned, the function $Q_{\alpha}(\xi_{AE}^{M_E})$, in analogy with QJSD, quantifies the distinguishability of the 
states in the ensemble. If the ensemble is generated by means of local measurements, its optimization over all local 
measurements quantifies classical correlations in the state $\rho_{AE}$. This concept is related to the channel capacity of a quantum - classical 
channel, where the capacity is rated by the HSW quantity \cite{holevo73,schumacher97}, that is the QJSD of the ensemble created by 
the quantum classical channel \cite{briet2009}. 
The following result provides a relation between the capacity of a quantum classical channel, 
a dephasing channel acting locally in a composed system, and the probability of success in discriminating the states in the output ensemble, depending 
on the entanglement with the purification ancillary system. 
It is considered that pure state $\psi_{ABE}$ is shared, and information, encoding on $E$, is sent from $A$  to 
$B$  by a quantum classical channel. 
This ensemble is created by means of the optimal local measurement over $E$,  
considering that there may be many copies of the state. 
\begin{teo}
Consider a pure state $\rho_{ABE}\in\M{D}(\B{C}_A\otimes\B{C}_B\otimes\B{C}_E)$, 
performing the optimal measurement over $E$ such that 
$C^{\infty}_{1/2}(\rho_{AE}) = \sup_{\xi_{AE}^{M_E}} Q_{\alpha}(\xi_{AE}^{M_E})$ 
and $\xi_{AE}^{M_E}  = \{ p_x, \rho_x^A\}$ is the ensemble 
in $A$ created by the local measurement. $C^{\infty}_{1/2}(\rho_{AE})$ is rated below as: 
\begin{equation}\label{thesao}
C^{\infty}_{1/2}(\rho_{AE}) \geq - \log{P_{suc}(X|A)} - LR_g^{\infty}(\rho_{AB}),
\end{equation}
where $P_{suc}(X|A)$ is the probability of success in discriminating the states in the ensemble $\xi_{AE}^{M_E}$, 
and $ LR_g^{\infty}(\rho_{AB})$ is asymptotic 
log generalized robustness .
\end{teo}
\begin{proof}
Given a regularized version of Eq.\eqref{t1} for $\alpha = 1/2$
\[
C^{\infty}_{1/2}(\rho_{AE}) =  S^{\infty}_{1/2}(\rho_A) - E^{\infty}_{1/2}(\rho_{AB}),
\]
where $f^{\infty}(\rho) = \lim_{n\rightarrow \infty} \frac{f(\rho^{\otimes n})}{n}$. 
Substituting Eq.\eqref{t.1.renner}, Eq.\eqref{lrgeof} and by linearity of the trace in definition of probability of success in Eq.\eqref{psuc}, it proves the statement. 
\end{proof}
Eq.\eqref{thesao} relates the character of a quantifier of distinguishability of $C_{1/2}$ with its correlation quantifier, 
relating it with the probability of 
success in discriminating the states in the ensemble $\xi_{AE}^{M_E}$ with quantum entanglement quantified by regularized LGR.

\section{Conclusion}
In this work a generalization of Koashi - Winter relation is presented  
by means of the Renyi\apos s entropic version of quantum 
Jensen Shannon divergence.  
From this generalization, some analytical results for quantifiers of classical and quantum correlations are presented. 
A $\alpha$ Renyi\apos s quantifier of classical correlations for $\alpha \in (0,1)$ is also introduced. 
As an application of the main result, a lower bound for $C^{\infty}_{1/2}(\rho_{AE})$ is obtained, 
related to the discrimination of the states in the ensemble, created by the local measurement, and 
the asymptotic log robustness of entanglement.
 This result expresses the character of the quantifier of distinguishability of the states in the ensemble composed of the measurement output states, 
in contrast with its character as a correlation quantifier. 

As a natural extension of these results one can define the Jensen Shannon divergence from the Sandwiched Renyi\apos s  relative entropy \cite{muller2013}
\[ QJSD_{\alpha} = D_{\alpha}(\rho_{AX}||\rho_A\otimes\rho_X), \]
where  
\[ D_{\alpha}(\rho||\sigma) = \frac{1}{\alpha - 1} \log \left\{  \Tr\left[ \left(\sigma^{\frac{1-\alpha}{2\alpha}}  \rho \sigma^{\frac{1-\alpha}{2\alpha}}\right)^{\alpha} \right] \right\}, \]
if $\text{supp}(\rho) \subseteq \text{supp}(\sigma)$, otherwise it is not finite. 
This is known to be monotone decreasing  \cite{hermes2017}. 
From this definition of QJSD one can study its relation with $\alpha$ entanglement of formation 
defined in Ref.\cite{wilde15,berta15b}, that is obtained from Sandwiched Renyi\apos s - relative entropy. 
Another interesting quantifier of correlations 
in this context is quantum discord, defined and explored in Ref.\cite{wilde15},  
obtained starting from the generalization of the Renyi\apos s conditional information \cite{berta15}. 
Some application in quantum information protocols remain to be explored \cite{mosonyi2015}, and related to the generalization of Koashi - Winter relations.

\acknowledgments
The author would like to thank F.F. Fanchini and the Infoquant Group (UFMG) for fruitful discussions. 

\section*{Conflict of Interest Disclosure} 
The author declares that there is no conflict of interest regarding the publication of this paper.


\begin{thebibliography}{64}
\expandafter\ifx\csname natexlab\endcsname\relax\def\natexlab#1{#1}\fi
\expandafter\ifx\csname bibnamefont\endcsname\relax
  \def\bibnamefont#1{#1}\fi
\expandafter\ifx\csname bibfnamefont\endcsname\relax
  \def\bibfnamefont#1{#1}\fi
\expandafter\ifx\csname citenamefont\endcsname\relax
  \def\citenamefont#1{#1}\fi
\expandafter\ifx\csname url\endcsname\relax
  \def\url#1{\texttt{#1}}\fi
\expandafter\ifx\csname urlprefix\endcsname\relax\def\urlprefix{URL }\fi
\providecommand{\bibinfo}[2]{#2}
\providecommand{\eprint}[2][]{\url{#2}}

\bibitem[{\citenamefont{Abeyesinghe et~al.}(2009)\citenamefont{Abeyesinghe,
  Devetak, Hayden, and Winter}}]{mother}
\bibinfo{author}{\bibfnamefont{A.}~\bibnamefont{Abeyesinghe}},
  \bibinfo{author}{\bibfnamefont{I.}~\bibnamefont{Devetak}},
  \bibinfo{author}{\bibfnamefont{P.}~\bibnamefont{Hayden}}, \bibnamefont{and}
  \bibinfo{author}{\bibfnamefont{A.}~\bibnamefont{Winter}},
  \bibinfo{journal}{Proceedings of the Royal Society A: Mathematical, Physical
  and Engineering Science} \textbf{\bibinfo{volume}{465}},
  \bibinfo{pages}{2537} (\bibinfo{year}{2009}).

\bibitem[{\citenamefont{Wilde}(2013)}]{wildebook}
\bibinfo{author}{\bibfnamefont{M.~M.} \bibnamefont{Wilde}},
  \emph{\bibinfo{title}{{Quantum Information Theory}}}
  (\bibinfo{publisher}{Cambridge University Press}, \bibinfo{year}{2013}), ISBN
  \bibinfo{isbn}{9781107034259},
  \urlprefix\url{https://books.google.com.br/books?id=T36v2Sp7DnIC}.

\bibitem[{\citenamefont{Tsallis}(1988)}]{tsallis88}
\bibinfo{author}{\bibfnamefont{C.}~\bibnamefont{Tsallis}},
  \bibinfo{journal}{Journal of statistical physics}
  \textbf{\bibinfo{volume}{52}}, \bibinfo{pages}{479} (\bibinfo{year}{1988}).

\bibitem[{\citenamefont{R{\'e}nyi et~al.}(1961)}]{renyi61}
\bibinfo{author}{\bibfnamefont{A.}~\bibnamefont{R{\'e}nyi}}
  \bibnamefont{et~al.}, in \emph{\bibinfo{booktitle}{{Proceedings of the fourth
  Berkeley symposium on mathematical statistics and probability}}}
  (\bibinfo{year}{1961}), vol.~\bibinfo{volume}{1}, pp.
  \bibinfo{pages}{547--561}.

\bibitem[{\citenamefont{Mosonyi and Hiai}(2011)}]{mosonyi2011}
\bibinfo{author}{\bibfnamefont{M.}~\bibnamefont{Mosonyi}} \bibnamefont{and}
  \bibinfo{author}{\bibfnamefont{F.}~\bibnamefont{Hiai}},
  \bibinfo{journal}{IEEE Transactions on Information Theory}
  \textbf{\bibinfo{volume}{57}}, \bibinfo{pages}{2474} (\bibinfo{year}{2011}).

\bibitem[{\citenamefont{Teixeira et~al.}(2012)\citenamefont{Teixeira, Matos,
  and Antunes}}]{matos12}
\bibinfo{author}{\bibfnamefont{A.}~\bibnamefont{Teixeira}},
  \bibinfo{author}{\bibfnamefont{A.}~\bibnamefont{Matos}}, \bibnamefont{and}
  \bibinfo{author}{\bibfnamefont{L.}~\bibnamefont{Antunes}},
  \bibinfo{journal}{IEEE Transactions on Information Theory}
  \textbf{\bibinfo{volume}{58}}, \bibinfo{pages}{4273} (\bibinfo{year}{2012}),
  ISSN \bibinfo{issn}{0018-9448}.

\bibitem[{\citenamefont{M{\"u}ller-Lennert
  et~al.}(2013)\citenamefont{M{\"u}ller-Lennert, Dupuis, Szehr, Fehr, and
  Tomamichel}}]{muller2013}
\bibinfo{author}{\bibfnamefont{M.}~\bibnamefont{M{\"u}ller-Lennert}},
  \bibinfo{author}{\bibfnamefont{F.}~\bibnamefont{Dupuis}},
  \bibinfo{author}{\bibfnamefont{O.}~\bibnamefont{Szehr}},
  \bibinfo{author}{\bibfnamefont{S.}~\bibnamefont{Fehr}}, \bibnamefont{and}
  \bibinfo{author}{\bibfnamefont{M.}~\bibnamefont{Tomamichel}},
  \bibinfo{journal}{Journal of Mathematical Physics}
  \textbf{\bibinfo{volume}{54}}, \bibinfo{pages}{122203}
  (\bibinfo{year}{2013}).

\bibitem[{\citenamefont{Fehr and Berens}(2014)}]{fehr2014}
\bibinfo{author}{\bibfnamefont{S.}~\bibnamefont{Fehr}} \bibnamefont{and}
  \bibinfo{author}{\bibfnamefont{S.}~\bibnamefont{Berens}},
  \bibinfo{journal}{IEEE Transactions on Information Theory}
  \textbf{\bibinfo{volume}{60}}, \bibinfo{pages}{6801} (\bibinfo{year}{2014}).

\bibitem[{\citenamefont{Mosonyi and Ogawa}(2015)}]{mosonyi2015}
\bibinfo{author}{\bibfnamefont{M.}~\bibnamefont{Mosonyi}} \bibnamefont{and}
  \bibinfo{author}{\bibfnamefont{T.}~\bibnamefont{Ogawa}},
  \bibinfo{journal}{Communications in Mathematical Physics}
  \textbf{\bibinfo{volume}{334}}, \bibinfo{pages}{1617} (\bibinfo{year}{2015}),
  ISSN \bibinfo{issn}{1432-0916},
  \urlprefix\url{http://dx.doi.org/10.1007/s00220-014-2248-x}.

\bibitem[{\citenamefont{Datta and Wilde}(2015)}]{datta15}
\bibinfo{author}{\bibfnamefont{N.}~\bibnamefont{Datta}} \bibnamefont{and}
  \bibinfo{author}{\bibfnamefont{M.~M.} \bibnamefont{Wilde}},
  \bibinfo{journal}{Journal of Physics A: Mathematical and Theoretical}
  \textbf{\bibinfo{volume}{48}}, \bibinfo{pages}{505301}
  (\bibinfo{year}{2015}).

\bibitem[{\citenamefont{Berta et~al.}(2015{\natexlab{a}})\citenamefont{Berta,
  Seshadreesan, and Wilde}}]{berta15}
\bibinfo{author}{\bibfnamefont{M.}~\bibnamefont{Berta}},
  \bibinfo{author}{\bibfnamefont{K.~P.} \bibnamefont{Seshadreesan}},
  \bibnamefont{and} \bibinfo{author}{\bibfnamefont{M.~M.} \bibnamefont{Wilde}},
  \bibinfo{journal}{Journal of Mathematical Physics}
  \textbf{\bibinfo{volume}{56}}, \bibinfo{pages}{022205}
  (\bibinfo{year}{2015}{\natexlab{a}}).

\bibitem[{\citenamefont{Berta et~al.}(2015{\natexlab{b}})\citenamefont{Berta,
  Seshadreesan, and Wilde}}]{berta15b}
\bibinfo{author}{\bibfnamefont{M.}~\bibnamefont{Berta}},
  \bibinfo{author}{\bibfnamefont{K.~P.} \bibnamefont{Seshadreesan}},
  \bibnamefont{and} \bibinfo{author}{\bibfnamefont{M.~M.} \bibnamefont{Wilde}},
  \bibinfo{journal}{Phys. Rev. A} \textbf{\bibinfo{volume}{91}},
  \bibinfo{pages}{022333} (\bibinfo{year}{2015}{\natexlab{b}}),
  \urlprefix\url{https://link.aps.org/doi/10.1103/PhysRevA.91.022333}.

\bibitem[{\citenamefont{Seshadreesan et~al.}(2015)\citenamefont{Seshadreesan,
  Berta, and Wilde}}]{wilde15}
\bibinfo{author}{\bibfnamefont{K.~P.} \bibnamefont{Seshadreesan}},
  \bibinfo{author}{\bibfnamefont{M.}~\bibnamefont{Berta}}, \bibnamefont{and}
  \bibinfo{author}{\bibfnamefont{M.~M.} \bibnamefont{Wilde}},
  \bibinfo{journal}{Journal of Physics A: Mathematical and Theoretical}
  \textbf{\bibinfo{volume}{48}}, \bibinfo{pages}{395303}
  (\bibinfo{year}{2015}).

\bibitem[{\citenamefont{Dupuis and Wilde}(2016)}]{wilde16b}
\bibinfo{author}{\bibfnamefont{F.}~\bibnamefont{Dupuis}} \bibnamefont{and}
  \bibinfo{author}{\bibfnamefont{M.~M.} \bibnamefont{Wilde}},
  \bibinfo{journal}{Quantum Information Processing}
  \textbf{\bibinfo{volume}{15}}, \bibinfo{pages}{1309} (\bibinfo{year}{2016}).

\bibitem[{\citenamefont{Leditzky et~al.}(2016)\citenamefont{Leditzky, Wilde,
  and Datta}}]{wilde16}
\bibinfo{author}{\bibfnamefont{F.}~\bibnamefont{Leditzky}},
  \bibinfo{author}{\bibfnamefont{M.~M.} \bibnamefont{Wilde}}, \bibnamefont{and}
  \bibinfo{author}{\bibfnamefont{N.}~\bibnamefont{Datta}},
  \bibinfo{journal}{Journal of Mathematical Physics}
  \textbf{\bibinfo{volume}{57}}, \bibinfo{pages}{082202}
  (\bibinfo{year}{2016}).

\bibitem[{\citenamefont{M{\"u}ller-Hermes and Reeb}(2017)}]{hermes2017}
\bibinfo{author}{\bibfnamefont{A.}~\bibnamefont{M{\"u}ller-Hermes}}
  \bibnamefont{and} \bibinfo{author}{\bibfnamefont{D.}~\bibnamefont{Reeb}},
  \bibinfo{journal}{Annales Henri Poincar{\'e}} pp. \bibinfo{pages}{1--12}
  (\bibinfo{year}{2017}), ISSN \bibinfo{issn}{1424-0661},
  \urlprefix\url{http://dx.doi.org/10.1007/s00023-017-0550-9}.

\bibitem[{\citenamefont{Kim and Sanders}(2010)}]{sanders2010}
\bibinfo{author}{\bibfnamefont{J.~S.} \bibnamefont{Kim}} \bibnamefont{and}
  \bibinfo{author}{\bibfnamefont{B.~C.} \bibnamefont{Sanders}},
  \bibinfo{journal}{Journal of Physics A: Mathematical and Theoretical}
  \textbf{\bibinfo{volume}{43}}, \bibinfo{pages}{445305}
  (\bibinfo{year}{2010}),
  \urlprefix\url{http://stacks.iop.org/1751-8121/43/i=44/a=445305}.

\bibitem[{\citenamefont{Bengtsson and Zyczkowski}(2006)}]{zyc}
\bibinfo{author}{\bibfnamefont{I.}~\bibnamefont{Bengtsson}} \bibnamefont{and}
  \bibinfo{author}{\bibfnamefont{K.}~\bibnamefont{Zyczkowski}},
  \emph{\bibinfo{title}{{Geometry of Quantum States}}}
  (\bibinfo{publisher}{Cambridge University Press}, \bibinfo{year}{2006}),
  \urlprefix\url{http://dx.doi.org/10.1017/CBO9780511535048}.

\bibitem[{\citenamefont{Neumark}(1940)}]{naimark}
\bibinfo{author}{\bibfnamefont{M.}~\bibnamefont{Neumark}},
  \bibinfo{journal}{Izvestiya Rossiiskoi Akademii Nauk. Seriya
  Matematicheskaya} \textbf{\bibinfo{volume}{4}}, \bibinfo{pages}{277}
  (\bibinfo{year}{1940}).

\bibitem[{\citenamefont{Wehrl}(1978)}]{wehrl78}
\bibinfo{author}{\bibfnamefont{A.}~\bibnamefont{Wehrl}}, \bibinfo{journal}{Rev.
  Mod. Phys.} \textbf{\bibinfo{volume}{50}}, \bibinfo{pages}{221}
  (\bibinfo{year}{1978}),
  \urlprefix\url{https://link.aps.org/doi/10.1103/RevModPhys.50.221}.

\bibitem[{\citenamefont{Ben-Bassat and Raviv}(1978)}]{raviv78}
\bibinfo{author}{\bibfnamefont{M.}~\bibnamefont{Ben-Bassat}} \bibnamefont{and}
  \bibinfo{author}{\bibfnamefont{J.}~\bibnamefont{Raviv}},
  \bibinfo{journal}{IEEE Transactions on Information Theory}
  \textbf{\bibinfo{volume}{24}}, \bibinfo{pages}{324} (\bibinfo{year}{1978}),
  ISSN \bibinfo{issn}{0018-9448}.

\bibitem[{\citenamefont{Majtey et~al.}(2005)\citenamefont{Majtey, Lamberti, and
  Prato}}]{majteyjensen}
\bibinfo{author}{\bibfnamefont{A.}~\bibnamefont{Majtey}},
  \bibinfo{author}{\bibfnamefont{P.}~\bibnamefont{Lamberti}}, \bibnamefont{and}
  \bibinfo{author}{\bibfnamefont{D.}~\bibnamefont{Prato}},
  \bibinfo{journal}{Physical Review A} \textbf{\bibinfo{volume}{72}},
  \bibinfo{pages}{052310} (\bibinfo{year}{2005}).

\bibitem[{\citenamefont{Lin}(1991)}]{jensenentropy}
\bibinfo{author}{\bibfnamefont{J.}~\bibnamefont{Lin}},
  \bibinfo{journal}{Information Theory, IEEE Transactions on}
  \textbf{\bibinfo{volume}{37}}, \bibinfo{pages}{145} (\bibinfo{year}{1991}).

\bibitem[{\citenamefont{Lamberti et~al.}(2008)\citenamefont{Lamberti, Majtey,
  Borras, Casas, and Plastino}}]{jensenmetric}
\bibinfo{author}{\bibfnamefont{P.}~\bibnamefont{Lamberti}},
  \bibinfo{author}{\bibfnamefont{A.}~\bibnamefont{Majtey}},
  \bibinfo{author}{\bibfnamefont{A.}~\bibnamefont{Borras}},
  \bibinfo{author}{\bibfnamefont{M.}~\bibnamefont{Casas}}, \bibnamefont{and}
  \bibinfo{author}{\bibfnamefont{A.}~\bibnamefont{Plastino}},
  \bibinfo{journal}{Physical Review A} \textbf{\bibinfo{volume}{77}},
  \bibinfo{pages}{052311} (\bibinfo{year}{2008}).

\bibitem[{\citenamefont{Holevo}(1973{\natexlab{a}})}]{holevo1973information}
\bibinfo{author}{\bibfnamefont{A.~S.} \bibnamefont{Holevo}},
  \bibinfo{journal}{Problemy Peredachi Informatsii}
  \textbf{\bibinfo{volume}{9}}, \bibinfo{pages}{31}
  (\bibinfo{year}{1973}{\natexlab{a}}).

\bibitem[{\citenamefont{Schumacher and Westmoreland}(1997)}]{schumacher97}
\bibinfo{author}{\bibfnamefont{B.}~\bibnamefont{Schumacher}} \bibnamefont{and}
  \bibinfo{author}{\bibfnamefont{M.~D.} \bibnamefont{Westmoreland}},
  \bibinfo{journal}{Physical Review A} \textbf{\bibinfo{volume}{56}},
  \bibinfo{pages}{131} (\bibinfo{year}{1997}).

\bibitem[{\citenamefont{Oppenheim et~al.}(2002)\citenamefont{Oppenheim,
  Horodecki, Horodecki, and Horodecki}}]{opp}
\bibinfo{author}{\bibfnamefont{J.}~\bibnamefont{Oppenheim}},
  \bibinfo{author}{\bibfnamefont{M.}~\bibnamefont{Horodecki}},
  \bibinfo{author}{\bibfnamefont{P.}~\bibnamefont{Horodecki}},
  \bibnamefont{and}
  \bibinfo{author}{\bibfnamefont{R.}~\bibnamefont{Horodecki}},
  \bibinfo{journal}{Phys. Rev. Lett.} \textbf{\bibinfo{volume}{89}},
  \bibinfo{pages}{180402} (\bibinfo{year}{2002}),
  \urlprefix\url{http://link.aps.org/doi/10.1103/PhysRevLett.89.180402}.

\bibitem[{\citenamefont{Piani et~al.}(2008)\citenamefont{Piani, Horodecki, and
  Horodecki}}]{piani08}
\bibinfo{author}{\bibfnamefont{M.}~\bibnamefont{Piani}},
  \bibinfo{author}{\bibfnamefont{P.}~\bibnamefont{Horodecki}},
  \bibnamefont{and}
  \bibinfo{author}{\bibfnamefont{R.}~\bibnamefont{Horodecki}},
  \bibinfo{journal}{Phys. Rev. Lett.} \textbf{\bibinfo{volume}{100}},
  \bibinfo{pages}{090502} (\bibinfo{year}{2008}),
  \urlprefix\url{http://link.aps.org/doi/10.1103/PhysRevLett.100.090502}.

\bibitem[{\citenamefont{Ollivier and Zurek}(2001)}]{zurek2001}
\bibinfo{author}{\bibfnamefont{H.}~\bibnamefont{Ollivier}} \bibnamefont{and}
  \bibinfo{author}{\bibfnamefont{W.~H.} \bibnamefont{Zurek}},
  \bibinfo{journal}{Phys. Rev. Lett.} \textbf{\bibinfo{volume}{88}},
  \bibinfo{pages}{017901} (\bibinfo{year}{2001}),
  \urlprefix\url{http://link.aps.org/doi/10.1103/PhysRevLett.88.017901}.

\bibitem[{\citenamefont{Henderson and Vedral}(2001)}]{henderson}
\bibinfo{author}{\bibfnamefont{L.}~\bibnamefont{Henderson}} \bibnamefont{and}
  \bibinfo{author}{\bibfnamefont{V.}~\bibnamefont{Vedral}},
  \bibinfo{journal}{Journal of Physics A: Mathematical and General}
  \textbf{\bibinfo{volume}{34}}, \bibinfo{pages}{6899} (\bibinfo{year}{2001}),
  \urlprefix\url{http://stacks.iop.org/0305-4470/34/i=35/a=315}.

\bibitem[{\citenamefont{F. et~al.}(2012)\citenamefont{F., K., Cornelio, and
  de~C.}}]{fanchini2012}
\bibinfo{author}{\bibfnamefont{F.~F.} \bibnamefont{F.}},
  \bibinfo{author}{\bibfnamefont{C.~L.} \bibnamefont{K.}},
  \bibinfo{author}{\bibfnamefont{M.~F.} \bibnamefont{Cornelio}},
  \bibnamefont{and} \bibinfo{author}{\bibfnamefont{O.~M.} \bibnamefont{de~C.}},
  \bibinfo{journal}{New Journal of Physics} \textbf{\bibinfo{volume}{14}},
  \bibinfo{pages}{013027} (\bibinfo{year}{2012}),
  \urlprefix\url{http://stacks.iop.org/1367-2630/14/i=1/a=013027}.

\bibitem[{\citenamefont{Werner}(1989)}]{werner89}
\bibinfo{author}{\bibfnamefont{R.~F.} \bibnamefont{Werner}},
  \bibinfo{journal}{Physical Review A} \textbf{\bibinfo{volume}{40}},
  \bibinfo{pages}{4277} (\bibinfo{year}{1989}).

\bibitem[{\citenamefont{Schr{\"o}dinger}(1935)}]{schrodinger35}
\bibinfo{author}{\bibfnamefont{E.}~\bibnamefont{Schr{\"o}dinger}},
  \bibinfo{journal}{Mathematical Proceedings of the Cambridge Philosophical
  Society} \textbf{\bibinfo{volume}{31}}, \bibinfo{pages}{555}
  (\bibinfo{year}{1935}).

\bibitem[{\citenamefont{Bennett et~al.}(1996)\citenamefont{Bennett, DiVincenzo,
  Smolin, and Wootters}}]{bennettpra}
\bibinfo{author}{\bibfnamefont{C.~H.} \bibnamefont{Bennett}},
  \bibinfo{author}{\bibfnamefont{D.~P.} \bibnamefont{DiVincenzo}},
  \bibinfo{author}{\bibfnamefont{J.~A.} \bibnamefont{Smolin}},
  \bibnamefont{and} \bibinfo{author}{\bibfnamefont{W.~K.}
  \bibnamefont{Wootters}}, \bibinfo{journal}{Phys. Rev. A}
  \textbf{\bibinfo{volume}{54}}, \bibinfo{pages}{3824} (\bibinfo{year}{1996}),
  \urlprefix\url{http://link.aps.org/doi/10.1103/PhysRevA.54.3824}.

\bibitem[{\citenamefont{Schumacher}(1995)}]{schumacher95}
\bibinfo{author}{\bibfnamefont{B.}~\bibnamefont{Schumacher}},
  \bibinfo{journal}{Phys. Rev. A} \textbf{\bibinfo{volume}{51}},
  \bibinfo{pages}{2738} (\bibinfo{year}{1995}),
  \urlprefix\url{http://link.aps.org/doi/10.1103/PhysRevA.51.2738}.

\bibitem[{\citenamefont{Ekert and Knight}(1995)}]{ekert95}
\bibinfo{author}{\bibfnamefont{A.}~\bibnamefont{Ekert}} \bibnamefont{and}
  \bibinfo{author}{\bibfnamefont{P.~L.} \bibnamefont{Knight}},
  \bibinfo{journal}{American Journal of Physics} \textbf{\bibinfo{volume}{63}},
  \bibinfo{pages}{415} (\bibinfo{year}{1995}).

\bibitem[{\citenamefont{Fanchini et~al.}(2011)\citenamefont{Fanchini, Cornelio,
  de~Oliveira, and Caldeira}}]{fanchini11}
\bibinfo{author}{\bibfnamefont{F.~F.} \bibnamefont{Fanchini}},
  \bibinfo{author}{\bibfnamefont{M.~F.} \bibnamefont{Cornelio}},
  \bibinfo{author}{\bibfnamefont{M.~C.} \bibnamefont{de~Oliveira}},
  \bibnamefont{and} \bibinfo{author}{\bibfnamefont{A.~O.}
  \bibnamefont{Caldeira}}, \bibinfo{journal}{Phys. Rev. A}
  \textbf{\bibinfo{volume}{84}}, \bibinfo{pages}{012313}
  (\bibinfo{year}{2011}),
  \urlprefix\url{http://link.aps.org/doi/10.1103/PhysRevA.84.012313}.

\bibitem[{\citenamefont{Horodecki et~al.}(1996)\citenamefont{Horodecki,
  Horodecki, and Horodecki}}]{horodecki96}
\bibinfo{author}{\bibfnamefont{R.}~\bibnamefont{Horodecki}},
  \bibinfo{author}{\bibfnamefont{P.}~\bibnamefont{Horodecki}},
  \bibnamefont{and}
  \bibinfo{author}{\bibfnamefont{M.}~\bibnamefont{Horodecki}},
  \bibinfo{journal}{Physics Letters A} \textbf{\bibinfo{volume}{210}},
  \bibinfo{pages}{377} (\bibinfo{year}{1996}), ISSN \bibinfo{issn}{0375-9601},
  \urlprefix\url{http://www.sciencedirect.com/science/article/pii/0375960195009302}.

\bibitem[{\citenamefont{Vidal}(2000)}]{vidal99}
\bibinfo{author}{\bibfnamefont{G.}~\bibnamefont{Vidal}},
  \bibinfo{journal}{Journal of Modern Optics} \textbf{\bibinfo{volume}{47}},
  \bibinfo{pages}{355} (\bibinfo{year}{2000}),
  \urlprefix\url{http://www.tandfonline.com/doi/abs/10.1080/09500340008244048}.

\bibitem[{\citenamefont{Koashi and Winter}(2004)}]{kw04}
\bibinfo{author}{\bibfnamefont{M.}~\bibnamefont{Koashi}} \bibnamefont{and}
  \bibinfo{author}{\bibfnamefont{A.}~\bibnamefont{Winter}},
  \bibinfo{journal}{Phys. Rev. A} \textbf{\bibinfo{volume}{69}},
  \bibinfo{pages}{022309} (\bibinfo{year}{2004}),
  \urlprefix\url{http://link.aps.org/doi/10.1103/PhysRevA.69.022309}.

\bibitem[{\citenamefont{Coffman et~al.}(2000)\citenamefont{Coffman, Kundu, and
  Wootters}}]{wooters00}
\bibinfo{author}{\bibfnamefont{V.}~\bibnamefont{Coffman}},
  \bibinfo{author}{\bibfnamefont{J.}~\bibnamefont{Kundu}}, \bibnamefont{and}
  \bibinfo{author}{\bibfnamefont{W.~K.} \bibnamefont{Wootters}},
  \bibinfo{journal}{Phys. Rev. A} \textbf{\bibinfo{volume}{61}},
  \bibinfo{pages}{052306} (\bibinfo{year}{2000}),
  \urlprefix\url{http://link.aps.org/doi/10.1103/PhysRevA.61.052306}.

\bibitem[{\citenamefont{Cornelio et~al.}(2011)\citenamefont{Cornelio,
  de~Oliveira, and Fanchini}}]{cornelio11}
\bibinfo{author}{\bibfnamefont{M.~F.} \bibnamefont{Cornelio}},
  \bibinfo{author}{\bibfnamefont{M.~C.} \bibnamefont{de~Oliveira}},
  \bibnamefont{and} \bibinfo{author}{\bibfnamefont{F.~F.}
  \bibnamefont{Fanchini}}, \bibinfo{journal}{Phys. Rev. Lett.}
  \textbf{\bibinfo{volume}{107}}, \bibinfo{pages}{020502}
  (\bibinfo{year}{2011}),
  \urlprefix\url{http://link.aps.org/doi/10.1103/PhysRevLett.107.020502}.

\bibitem[{\citenamefont{Cavalcanti et~al.}(2011)\citenamefont{Cavalcanti,
  Aolita, Boixo, Modi, Piani, and Winter}}]{esm}
\bibinfo{author}{\bibfnamefont{D.}~\bibnamefont{Cavalcanti}},
  \bibinfo{author}{\bibfnamefont{L.}~\bibnamefont{Aolita}},
  \bibinfo{author}{\bibfnamefont{S.}~\bibnamefont{Boixo}},
  \bibinfo{author}{\bibfnamefont{K.}~\bibnamefont{Modi}},
  \bibinfo{author}{\bibfnamefont{M.}~\bibnamefont{Piani}}, \bibnamefont{and}
  \bibinfo{author}{\bibfnamefont{A.}~\bibnamefont{Winter}},
  \bibinfo{journal}{Phys. Rev. A} \textbf{\bibinfo{volume}{83}},
  \bibinfo{pages}{032324} (\bibinfo{year}{2011}),
  \urlprefix\url{http://link.aps.org/doi/10.1103/PhysRevA.83.032324}.

\bibitem[{\citenamefont{Madhok and Datta}(2011)}]{dattasm}
\bibinfo{author}{\bibfnamefont{V.}~\bibnamefont{Madhok}} \bibnamefont{and}
  \bibinfo{author}{\bibfnamefont{A.}~\bibnamefont{Datta}},
  \bibinfo{journal}{Phys. Rev. A} \textbf{\bibinfo{volume}{83}},
  \bibinfo{pages}{032323} (\bibinfo{year}{2011}),
  \urlprefix\url{http://link.aps.org/doi/10.1103/PhysRevA.83.032323}.

\bibitem[{\citenamefont{Yu et~al.}(2011)\citenamefont{Yu, Zhang, Chen, and
  Oh}}]{yu}
\bibinfo{author}{\bibfnamefont{S.}~\bibnamefont{Yu}},
  \bibinfo{author}{\bibfnamefont{C.}~\bibnamefont{Zhang}},
  \bibinfo{author}{\bibfnamefont{Q.}~\bibnamefont{Chen}}, \bibnamefont{and}
  \bibinfo{author}{\bibfnamefont{C.}~\bibnamefont{Oh}},
  \bibinfo{journal}{arXiv:1102.1301v2}  (\bibinfo{year}{2011}),
  \urlprefix\url{http://arxiv.org/abs/1102.1301v2}.

\bibitem[{\citenamefont{Xi et~al.}(2012)\citenamefont{Xi, Lu, Wang, and
  Li}}]{xi}
\bibinfo{author}{\bibfnamefont{Z.}~\bibnamefont{Xi}},
  \bibinfo{author}{\bibfnamefont{X.-M.} \bibnamefont{Lu}},
  \bibinfo{author}{\bibfnamefont{X.}~\bibnamefont{Wang}}, \bibnamefont{and}
  \bibinfo{author}{\bibfnamefont{Y.}~\bibnamefont{Li}}, \bibinfo{journal}{Phys.
  Rev. A} \textbf{\bibinfo{volume}{85}}, \bibinfo{pages}{032109}
  (\bibinfo{year}{2012}),
  \urlprefix\url{https://link.aps.org/doi/10.1103/PhysRevA.85.032109}.

\bibitem[{\citenamefont{Xi et~al.}(2011)\citenamefont{Xi, Lu, Wang, and
  Li}}]{xib}
\bibinfo{author}{\bibfnamefont{Z.}~\bibnamefont{Xi}},
  \bibinfo{author}{\bibfnamefont{X.-M.} \bibnamefont{Lu}},
  \bibinfo{author}{\bibfnamefont{X.}~\bibnamefont{Wang}}, \bibnamefont{and}
  \bibinfo{author}{\bibfnamefont{Y.}~\bibnamefont{Li}},
  \bibinfo{journal}{Journal of Physics A: Mathematical and Theoretical}
  \textbf{\bibinfo{volume}{44}}, \bibinfo{pages}{375301}
  (\bibinfo{year}{2011}),
  \urlprefix\url{http://stacks.iop.org/1751-8121/44/i=37/a=375301}.

\bibitem[{\citenamefont{Zhang et~al.}(2011)\citenamefont{Zhang, Yu, Chen, and
  Oh}}]{zhang}
\bibinfo{author}{\bibfnamefont{C.}~\bibnamefont{Zhang}},
  \bibinfo{author}{\bibfnamefont{S.}~\bibnamefont{Yu}},
  \bibinfo{author}{\bibfnamefont{Q.}~\bibnamefont{Chen}}, \bibnamefont{and}
  \bibinfo{author}{\bibfnamefont{C.~H.} \bibnamefont{Oh}},
  \bibinfo{journal}{Phys. Rev. A} \textbf{\bibinfo{volume}{84}},
  \bibinfo{pages}{052112} (\bibinfo{year}{2011}),
  \urlprefix\url{http://link.aps.org/doi/10.1103/PhysRevA.84.052112}.

\bibitem[{\citenamefont{Chi and Lee}(2003)}]{lastra}
\bibinfo{author}{\bibfnamefont{D.~P.} \bibnamefont{Chi}} \bibnamefont{and}
  \bibinfo{author}{\bibfnamefont{S.}~\bibnamefont{Lee}},
  \bibinfo{journal}{Journal of Physics A: Mathematical and General}
  \textbf{\bibinfo{volume}{36}}, \bibinfo{pages}{11503} (\bibinfo{year}{2003}),
  \urlprefix\url{http://stacks.iop.org/0305-4470/36/i=45/a=010}.

\bibitem[{\citenamefont{Cen et~al.}(2011)\citenamefont{Cen, Li, Shao, and
  Yan}}]{cen}
\bibinfo{author}{\bibfnamefont{L.-X.} \bibnamefont{Cen}},
  \bibinfo{author}{\bibfnamefont{X.-Q.} \bibnamefont{Li}},
  \bibinfo{author}{\bibfnamefont{J.}~\bibnamefont{Shao}}, \bibnamefont{and}
  \bibinfo{author}{\bibfnamefont{Y.}~\bibnamefont{Yan}},
  \bibinfo{journal}{Phys. Rev. A} \textbf{\bibinfo{volume}{83}},
  \bibinfo{pages}{054101} (\bibinfo{year}{2011}),
  \urlprefix\url{http://link.aps.org/doi/10.1103/PhysRevA.83.054101}.

\bibitem[{\citenamefont{Bri{\"e}t and Harremo{\"e}s}(2009)}]{briet2009}
\bibinfo{author}{\bibfnamefont{J.}~\bibnamefont{Bri{\"e}t}} \bibnamefont{and}
  \bibinfo{author}{\bibfnamefont{P.}~\bibnamefont{Harremo{\"e}s}},
  \bibinfo{journal}{Physical review A} \textbf{\bibinfo{volume}{79}},
  \bibinfo{pages}{052311} (\bibinfo{year}{2009}).

\bibitem[{\citenamefont{Brodutch and Modi}(2012)}]{brodutch12}
\bibinfo{author}{\bibfnamefont{A.}~\bibnamefont{Brodutch}} \bibnamefont{and}
  \bibinfo{author}{\bibfnamefont{K.}~\bibnamefont{Modi}},
  \bibinfo{journal}{Quantum Information \& Computation}
  \textbf{\bibinfo{volume}{12}}, \bibinfo{pages}{721} (\bibinfo{year}{2012}).

\bibitem[{\citenamefont{Tomamichel}(2012)}]{tomathesis}
\bibinfo{author}{\bibfnamefont{M.}~\bibnamefont{Tomamichel}},
  \bibinfo{journal}{arXiv preprint arXiv:1203.2142}  (\bibinfo{year}{2012}).

\bibitem[{\citenamefont{Vidal and Tarrach}(1999)}]{vida99}
\bibinfo{author}{\bibfnamefont{G.}~\bibnamefont{Vidal}} \bibnamefont{and}
  \bibinfo{author}{\bibfnamefont{R.}~\bibnamefont{Tarrach}},
  \bibinfo{journal}{Physical Review A} \textbf{\bibinfo{volume}{59}},
  \bibinfo{pages}{141} (\bibinfo{year}{1999}).

\bibitem[{\citenamefont{Steiner}(2003)}]{steiner03}
\bibinfo{author}{\bibfnamefont{M.}~\bibnamefont{Steiner}},
  \bibinfo{journal}{Physical Review A} \textbf{\bibinfo{volume}{67}},
  \bibinfo{pages}{054305} (\bibinfo{year}{2003}).

\bibitem[{\citenamefont{Brand{\~a}o}(2005)}]{brandaopra2005}
\bibinfo{author}{\bibfnamefont{F.~G. S.~L.} \bibnamefont{Brand{\~a}o}},
  \bibinfo{journal}{Phys. Rev. A} \textbf{\bibinfo{volume}{72}},
  \bibinfo{pages}{022310} (\bibinfo{year}{2005}),
  \urlprefix\url{http://link.aps.org/doi/10.1103/PhysRevA.72.022310}.

\bibitem[{\citenamefont{Brand{\~a}o and Plenio}(2010)}]{brandao10}
\bibinfo{author}{\bibfnamefont{F.~G.} \bibnamefont{Brand{\~a}o}}
  \bibnamefont{and} \bibinfo{author}{\bibfnamefont{M.~B.}
  \bibnamefont{Plenio}}, \bibinfo{journal}{Communications in Mathematical
  Physics} \textbf{\bibinfo{volume}{295}}, \bibinfo{pages}{829}
  (\bibinfo{year}{2010}).

\bibitem[{\citenamefont{Brandao and Plenio}(2008)}]{bnature}
\bibinfo{author}{\bibfnamefont{F.~G.} \bibnamefont{Brandao}} \bibnamefont{and}
  \bibinfo{author}{\bibfnamefont{M.~B.} \bibnamefont{Plenio}},
  \bibinfo{journal}{Nature Physics} \textbf{\bibinfo{volume}{4}},
  \bibinfo{pages}{873} (\bibinfo{year}{2008}).

\bibitem[{\citenamefont{Datta}(2009)}]{datta09}
\bibinfo{author}{\bibfnamefont{N.}~\bibnamefont{Datta}},
  \bibinfo{journal}{International Journal of Quantum Information}
  \textbf{\bibinfo{volume}{7}}, \bibinfo{pages}{475} (\bibinfo{year}{2009}).

\bibitem[{\citenamefont{Petz}(1986)}]{petz86}
\bibinfo{author}{\bibfnamefont{D.}~\bibnamefont{Petz}},
  \bibinfo{journal}{Reports on mathematical physics}
  \textbf{\bibinfo{volume}{23}}, \bibinfo{pages}{57} (\bibinfo{year}{1986}).

\bibitem[{\citenamefont{Konig et~al.}(2009)\citenamefont{Konig, Renner, and
  Schaffner}}]{renner}
\bibinfo{author}{\bibfnamefont{R.}~\bibnamefont{Konig}},
  \bibinfo{author}{\bibfnamefont{R.}~\bibnamefont{Renner}}, \bibnamefont{and}
  \bibinfo{author}{\bibfnamefont{C.}~\bibnamefont{Schaffner}},
  \bibinfo{journal}{Information Theory, IEEE Transactions on}
  \textbf{\bibinfo{volume}{55}}, \bibinfo{pages}{4337} (\bibinfo{year}{2009}).

\bibitem[{\citenamefont{Renner}(2008)}]{rennerthesis}
\bibinfo{author}{\bibfnamefont{R.}~\bibnamefont{Renner}},
  \bibinfo{journal}{International Journal of Quantum Information}
  \textbf{\bibinfo{volume}{6}}, \bibinfo{pages}{1} (\bibinfo{year}{2008}).

\bibitem[{\citenamefont{Hayden et~al.}(2001)\citenamefont{Hayden, Horodecki,
  and Terhal}}]{hayden01}
\bibinfo{author}{\bibfnamefont{P.~M.} \bibnamefont{Hayden}},
  \bibinfo{author}{\bibfnamefont{M.}~\bibnamefont{Horodecki}},
  \bibnamefont{and} \bibinfo{author}{\bibfnamefont{B.~M.}
  \bibnamefont{Terhal}}, \bibinfo{journal}{Journal of Physics A: Mathematical
  and General} \textbf{\bibinfo{volume}{34}}, \bibinfo{pages}{6891}
  (\bibinfo{year}{2001}).

\bibitem[{\citenamefont{Holevo}(1973{\natexlab{b}})}]{holevo73}
\bibinfo{author}{\bibfnamefont{A.}~\bibnamefont{Holevo}},
  \bibinfo{journal}{Journal of Multivariate Analysis}
  \textbf{\bibinfo{volume}{3}}, \bibinfo{pages}{337}
  (\bibinfo{year}{1973}{\natexlab{b}}), ISSN \bibinfo{issn}{0047-259X},
  \urlprefix\url{http://www.sciencedirect.com/science/article/pii/0047259X73900286}.

\end{thebibliography}
\end{document}